\documentclass[onecolumn,draft, 11pt]{IEEEtran}

\usepackage{color}
\usepackage{enumitem}
\usepackage{graphicx,tabularx,array,geometry,amsmath,amsthm,thmtools}

\usepackage{mathtools}
\DeclarePairedDelimiter{\ceil}{\lceil}{\rceil}
\DeclarePairedDelimiter\floor{\lfloor}{\rfloor}

\usepackage{color}

\usepackage{algorithm}
\usepackage{algorithmic}
\usepackage{amsfonts}
\usepackage{bm}
\usepackage{bbm}
\usepackage{makecell}
\usepackage{multirow}
 \usepackage{amssymb}
\usepackage{txfonts}
\usepackage[T1]{fontenc}
\usepackage{tikz}
\usepackage[scr=dutchcal]{mathalfa}

\usepackage{cite}

\newtheorem{Theorem}{Theorem}
\newtheorem{Proposition}{Proposition}

\newtheorem{Corollary}{Corollary}
\newtheorem{Example}{Example}
\newtheorem{Remark}{Remark}
\newtheorem{Definition}{Definition}

\DeclareMathOperator*{\argmax}{argmax}

\begin{document}
\author{Shahram Shahsavari, Farhad Shirani  and Elza Erkip\\
Dept. of Electrical and Computer Engineering \\
New York University, NY. \\\date{} }

\title{On the Fundamental Limits of Multi-user Scheduling under Short-term Fairness Constraints}


\maketitle
\begin{abstract}
\let\thefootnote\relax\footnotetext{This work is supported by NYU WIRELESS Industrial Affiliates and National Science Foundation grants EARS-1547332 and NeTS-1527750.}
In the conventional information theoretic analysis of multiterminal communication scenarios, it is often assumed that all of the distributed terminals use the communication channel simultaneously. However, in practical wireless communication systems --- due to restricted computation complexity at network terminals --- a limited number of users can be activated either in uplink or downlink simultaneously. This necessitates the design of a scheduler which determines the set of active users at each time-slot. A well-designed scheduler maximizes the average system utility subject to a set of fairness criteria, which must be met in a limited window-length to avoid long starvation periods.
In this work, scheduling under short-term temporal fairness constraints is considered. The objective is to maximize the average system utility such that the fraction of the time-slots that each user is activated is within desired upper and lower bounds in the fairness window-length. 
The set of feasible window-lengths is characterized as a function of system parameters. It is shown that the optimal system utility is non-monotonic and super-additive in window-length. Furthermore, a scheduling strategy is proposed which satisfies short-term fairness constraints for arbitrary window-lengths, and achieves optimal average system utility as the window-length is increased asymptotically. Numerical simulations are provided  to verify the  results.
\end{abstract}

\section{Introduction}
The design of efficient resource allocation and scheduling schemes in cellular communications is a topic of significant interest \cite{asadi2013survey, issariyakul2004throughput,kulkarni2003opportunistic,rosenberg-short-term,liu-jsac,lu1999fair,arxiv}. Opportunistic scheduling strategies maximize system utility by exploiting the channel state information while guaranteeing satisfaction of fairness constraints --- such as temporal fairness --- for individual users in the network.
In practice, these fairness constraints often need to be satisfied over limited window-lengths to avoid long starving periods and control system latency \cite{rosenberg-short-term,lu1999fair}.

In cellular networks, each cell operates in the uplink or downlink mode of operation at any given time. In uplink communication, a subset of users transmit their messages to the base station (BS). This can be modeled as communication over a multiple-access channel (MAC). In downlink communication, the BS transmits messages to a subset of users. This can be modeled as a broadcast channel (BC) communication problem. While the traditional information theoretic approach considers optimal communication rates for a \textit{given set of users}, scheduling strategies determine \textit{which users} are activated by the BS at each transmission block. The choice is made in accordance with the user's fairness demands (e.g. temporal demands) as well as the resulting system utility (e.g. sum-rate). Once the choice of active users is taken, physical layer techniques developed for MAC and BC are used for communication over the resulting channel \cite{ahlswede1973multi,yu2004sum}. 

In this work, we study the design of optimal scheduling strategies under short-term temporal fairness constraints. Temporal fairness requires
the fraction of the resource blocks that each user is activated to be within desired upper and lower bounds. 
 Temporally fair schedulers provide each user with a minimum temporal share in order to control the average delay \cite{issariyakul2004throughput}. 
Additionally, the maximum power drain of users can be restricted by limiting their activation time through placing upper-bounds on their temporal shares \cite{kulkarni2003opportunistic}. There has been a significant body of work dedicated to the study of temporally fair single-user \cite{liu-infocom,shahram-letter,shahsavari2018joint} and multi-user \cite{arxiv} schedulers under long-term fairness constraints. 

Short-term temporal fairness, where the fairness criteria are required to be satisfied over  a limited scheduling window, has been considered in several prior works for single-user schedulers \cite{rosenberg-short-term,liu-jsac,lu1999fair}. 
Short-term temporal fairness avoids long waiting times for users with low average channel quality. 
The use of timers in the networking protocols (e.g. TCP) gives further merit to considering short-term temporal fairness as the timer may expire and cause a loss of connectivity absent short-term temporal fairness guarantees \cite{rosenberg-short-term}.  
In \cite{rosenberg-short-term}, it is argued that long-term scheduling policies do not provide fairness guarantees in scheduling windows with practical length. 

In this work, we provide a general formulation of multi-user scheduling under short-term temporal fairness constraints. We characterize the set of feasible window-lengths and user temporal shares as a function of system parameters.
We show that in general, the optimal utility function is non-monotonic and super-additive as a function of the window-length. We build upon our prior works on the design of threshold based strategies under long-term fairness constraints \cite{arxiv,arxiv_Allerton} to propose an implementable scheduler which guarantees short-term fairness. Concentration of measure tools such as typicality are used to analyze the resulting utility of the proposed strategy. We show that as the window-length increases, the utility approaches that of the optimal one under long-term temporal constraints.  Furthermore, we investigate bounds on the optimal system utility as a function of the window-length and derive rates of convergence to optimal long-term utility. 

The rest of the paper is organized as follows: Section \ref{Sec:Not} explains the notation used in the rest of the paper. Section \ref{Sec:Form} provides the problem formulation. Section \ref{Sec:Ex} investigates the monotonicity of the optimal system utility as a function of window-length. Section \ref{Sec:Scheme} includes the proposed scheduling strategy. Section \ref{Sec:Sim} provides various simulations of practical scenarios to verify the results. Section \ref{Sec:Conc} concludes the paper.  

\section{Notation}
\label{Sec:Not}
The set of natural numbers, rational numbers and the real numbers are shown by $\mathbb{N}$, $\mathbb{Q}$ and $\mathbb{R}$ respectively.  The set of numbers $\{1,2,\cdots, n\}, n\in \mathbb{N}$ is represented by $[n]$. For the numbers $n_1,n_2,\cdots, n_{k}\in \mathbb{N}$, the least common multiple is written as $lcm(n_1,n_2,\cdots, n_{k})$. We write $a|b$ if the integer $a$ is a divisor of $b$.
The vector $(x_1,x_2,\cdots, x_n)$ is written as $x^n$. The $m\times t$ matrix $[g_{i,j}]_{i\in [m], j\in [t]}$ is denoted by $g^{m\times t}$.
For a random variable $X$, the corresponding probability space is $(\mathcal{X}, \mathbf{F}_{X}, P_X)$, where $\mathbf{F}_X$ is the underlying $\sigma$-field. 
For an event $\mathcal{A}$, the random variable $\mathbbm{1}_{\mathcal{A}}$ is the indicator function. 
Families of sets are shown using sans-serif letters such as $\mathsf{X}$.   The notation $\floor{x}$ ($\ceil{x}$) is used to represent the closest integer smaller (bigger) than $x$.

\section{Problem Formulation}
\label{Sec:Form}

Opportunistic multi-user scheduling under long-term temporal fairness constraints was formulated in \cite{arxiv}. We provide a summary and formulate multi-user scheduling under short-term fairness constraints. We consider a cell consisting of $n$ users and one BS. Only specific subsets of users may be activated simultaneously. For instance, in practical wireless communication systems --- due to restricted computation complexity at network terminals --- a limited number of users can be activated either in uplink or downlink at each time-slot. Subsets of users which can be activated simultaneously are called \textit{virtual users}. The set of virtual users is denoted by $\mathsf{V}=\{\mathcal{V}_1,\mathcal{V}_2,\cdots,\mathcal{V}_m\}$ where $\mathcal{V}_j, j\in [m]$ are subsets of users and $m\leq 2^n$. The choice of the active virtual user at a given time-slot determines the resulting system utility at that time-slot. The vector of system utilities due to activating each of the virtual users is called the \textit{performance vector}. As an example, the elements of the performance vector may be the sum-rate resulting from activating the corresponding virtual user in the downlink BC or uplink MAC. In this case, the performance vector is random and its value depends on the realization of the underlying  time-varying channel. In a given time-slot, the system utilities due to activating different virtual users may depend on each other. The performance vectors in different time-slots are assumed to be independent of each other.

\begin{Definition}[\bf{Performance Vector}]
The vector of jointly continuous variables $(R_{1,t},R_{2,t}, \cdots, R_{m,t})$, $t\in \mathbb{N}$ is the performance vector of the virtual users at time $t$. The sequence $(R_{1,t},R_{2,t}, \cdots, R_{m,t})$ is a sequence of independent\footnote{Note that the realization of the performance vector is independent over time, however the performance of the virtual users at a given time-slot may be dependent on each other.} vectors distributed identically according to the joint density $f_{R^m}$.
\end{Definition}
 Under temporal fairness, it is required that the fraction of time-slots in which each user is activated is bounded from below (above). The vector of lower (upper) bounds $\underline{w}^n$ ($\overline{w}^n)$ is called the lower (upper) temporal demand vector. We assume that the temporal demand vectors take values among the rational numbers (i.e. $(\underline{w}^n,\overline{w}^n)\in \mathbb{Q}^{n}\times \mathbb{Q}^n$). This does not cause a loss of generality since the rational numbers are dense in $\mathbb{R}$. Furthermore, we write $\underline{w}_i=\frac{\underline{k}_i}{\underline{d}_i},  i\in \mathbb{N}$ and $\overline{w}_i=\frac{\overline{k}_i}{\overline{d}_i}, i\in \mathbb{N}$, where $(\underline{k_i}, \underline{d}_i)$ and $(\overline{k_i}, \overline{d}_i)$ are assumed to be coprime integers.

The objective is to design a scheduling strategy satisfying the temporal fairness constraints in a given window-length while maximizing the resulting system utility. Accordingly, a scheduling strategy is defined as follows.

\begin{Definition}[\bf{$s$-Scheduler}] \label{Def:Strategy}
Consider the scheduling setup parametrized by $(n,\mathsf{V}, \underline{w}^n, \overline{w}^n, f_{R^m})$. A scheduling strategy $Q= (Q_t)_{t\in [s]}$  with window-length $s\in\mathbb{N}$ ($s$-scheduler) is a family of (possibly stochastic) functions $Q_t: \mathbb{R}^{m\times t} \to \mathsf{V}, t\in [s]$, where:
\begin{itemize}[leftmargin=*]
    \item { The input to $Q_t, t\in[s]$ is the
matrix of performance vectors $R^{m\times t}$ which consists of $t$ independently and identically distributed column vectors with distribution} $f_{R^m}$.
\item{ The temporal demand constraints are satisfied:
\begin{align}
P\left(\underline{w}_i\leq {A}_{i,s}^{Q} \leq \overline{w}_i, i\in [n]\right)=1,
\label{Def:tem_fair}
\end{align}}
\end{itemize}
where, the temporal share of user $u_i, i\in [n]$ up to time $t\in [s]$ is defined as
\begin{align}
A^Q_{i,t}=\frac{1}{t}\sum^t_{k=1}\mathbbm{1}_{\big\{u_i\in Q_k(R^{m\times k})\big\}}, \forall i\in [n], t\in [s].
\label{Def:temp_share}
\end{align}
\end{Definition}

Unless otherwise stated, we consider homogeneous systems where the scheduler is allowed to activate subsets of at most $N_{max}$ users at each time-slot.  More precisely, for a homogeneous multi-user system with $n$ users and maximum number of active users $N_{max}\leq n$, the set of virtual users is defined as
\[\mathsf{V}= \left\{\mathcal{V}_j\subset \mathcal{U}\big| |\mathcal{V}_j|\leq N_{max}\right\}. \]
We write $(n,N_{max}, \underline{w}^n,$ $\overline{w}^n,f_{R^m})$ instead of $(n,\mathsf{V}, \underline{w}^n,$ $\overline{w}^n,f_{R^m})$ to characterize a 
homogeneous system.

A scheduling setup where the user temporal shares are required to take a specific value, i.e. ${A}_{i,s}^Q=w_i, i\in [n]$, is called a setup with \textit{equality temporal constraints} and is parametrized by $(n,N_{max}, {w}^n, {w}^n, f_{R^m})$. A necessary condition for the existence of a scheduling strategy is that $\sum_{i\in [n]}w_i \leq N_{max}$ must hold since at most $N_{max}$ users can be activated simultaneously.

Note that some temporal demand vectors $(\underline{w}^n,\overline{w}^n)$ are not feasible given a specific window-length $s$. For instance if the window-length is taken to be $s=2$, then each user can be activated for $0$,$1$ or $2$ time-slots. Consequently, only temporal shares in the set $\{0,\frac{1}{2},1\}$ are feasible. The feasibility is also affected by $N_{max}$, the maximum number of users which can be activated simultaneously since $\sum_{i\in [n]}w_i =\sum_{i\in [n]} \frac{k_i}{d_i}\leq N_{max}$, where $d_i|s$. The following defines the set of feasible window-lengths given a temporal demand vector and virtual user set.

\begin{Definition}{\bf (Feasible Scheduling Window)}
For a given temporal demand vector $(\underline{w}^n,\overline{w}^n)\in \mathbb{Q}^{n}\times \mathbb{Q}^n$ and virtual user set $\mathsf{V}$, the window-length $s$ is called feasible if a scheduling strategy satisfying the temporal demand constraints $(\underline{w}^n,\overline{w}^n)$ exists. The set of all feasible window-lengths is denoted by $\mathcal{S}(\underline{w}^n,\overline{w}^n,\mathsf{V})$. For a homogeneous systems, the set of feasible window-lengths is denoted by $\mathcal{S}(\underline{w}^n,\overline{w}^n,N_{max})$.
\end{Definition}
Consider the setup with equality temporal fairness constraints $(n,N_{max}, {w}^n, {w}^n, f_{R^m})$. Assume that the  window-length $s$ is feasible (i.e. $s\in \mathcal{S}(w^n,w^n,N_{max})$). Then, there exists a vector of virtual user temporal shares $a^m$ satisfying the fairness constraints, where $a_j, j\in [m]$ is the fraction of time-slots in which the $j$th virtual user is activated. This yields a mapping from $w^n\in \mathbb{R}^n$ to $a^m\in \mathbb{R}^m$ which is described below. 

\begin{Definition}
\label{Def:Mapping}
For a tuple $(s,{w}^n,N_{max})$ such that $s\in \mathcal{S}({w}^n,{w}^n,N_{max})$, let $a^m$ be a vector of virtual user temporal shares such that i) $ \sum_{j\in [m]} a_j= 1$, ii) $\sum_{j: u_i\in \mathcal{V}_j} a_{j}= w_i,  \forall i\in [n]$, and iii) $0\leq a_{j}, \forall j\in [m]$,
where $a_j= \frac{l_j}{s}, j\in [m]$. We define the mapping $\Theta$ as $\Theta: (s,{w}^n,N_{max}) \mapsto (a_{1},a_{2},\cdots,a_m)$. 
\end{Definition}

 We make particular use of a class of non-opportunistic scheduling strategies called the \textit{ordered round robin} in investigating the feasibility of scheduling window-lengths. The ordered round robin strategy is an extension of single-user round robin to the multi-user settings, in which the virtual users are activated in a specific order.

 \begin{Definition}[{\bf Ordered Round Robin}] \label{Def:RR} For the scheduling setup $(n,N_{max}, \underline{w}^n, \overline{w}^n, f_{R^m})$ with window-length $s\in \mathcal{S}({w}^n,{w}^n,N_{max})$, an ordered round robin strategy has the form $Q_t(R^{m\times t})= \mathcal{V}_{j_t}, t\in [s]$ where $j_t$ is the unique index for which the following inequality holds
 \begin{align}
    \sum_{j< j_t} sa_j +1 \leq t\leq \sum_{j\leq j_t} sa_j, \label{Eq:RR}
 \end{align}
 where $a^m= \Theta (s,w^n,N_{max})$ for some $w^n$ such that 
 $\underline{w}_i\leq w_i\leq \overline{w}_i, i\in [n]$.
 \end{Definition}

 The average system utility of an $s$-scheduler is defined as:

\begin{Definition}[\bf{System Utility}]
For an $s$-scheduler $Q$:
\begin{itemize}[wide=0pt]
    \item {The average system utility up to time t, is defined as 
\begin{align}
U^Q_t&=\frac{1}{t}\sum^t_{k=1}\sum_{j=1}^m R_{j,k}\mathbbm{1}_{\big\{Q_k(R^{m\times k})=\mathcal{V}_{j}\big\}}.
\label{Def:sys_utility}
\end{align}}
\item{The variable $U^Q_s$ is called the average system utility for the $s$-scheduler.}  An $s$-scheduler $Q_s^*$ is optimal if and only if 
$Q_s^*\in\argmax_{Q\in \mathcal{Q}_s} U^Q_{s}$,
where $\mathcal{Q}_s$ is the set of all $s$-schedulers for the scheduling setup. The optimal utility is denoted by $U^*_s$.
\end{itemize}
\end{Definition}
Our objective is to study properties of $U^*_s$ and to design scheduling strategies satisfying the short-term fairness constraints while achieving average system utilities close to $U^*_s$.

\section{Characteristics of Short-term Schedulers}
\label{Sec:Ex}
In this section, we study the set of feasible window-lengths, and show that under equality temporal constraints, the set is non-contiguous, whereas, under inequality temporal constraints, the set is asymptotically contiguous. Consequently, under inequality constraints, for any given temporal demand vectors and a set of virtual users, there is a minimum window-length such that all larger window-lengths are feasible. This allows us to investigate the optimal utility function as the window-length is increased asymptotically. 
Furthermore, we show that the optimal utility function $U^*_s$ is a non-monotonic and supper-additive function of the window-length $s$. 

\subsection{Contiguity of the Feasible Window-length Set}
As a first step, we show an equivalence relation between the feasibility of a given window-length and the existence of an ordered round robin strategy with that window-length. This leads to a significant simplification of the proofs provided in the rest of the section. 

\begin{Proposition}
For an scheduling setup $(n,N_{max}, \underline{w}^n, \overline{w}^n, f_{R^m})$, the window-length $s$ is feasible iff there exists an ordered round robin strategy satisfying the temporal constraints. 
\end{Proposition}
\begin{proof}
Assume that the window-length is feasible, then by definition a strategy satisfying the temporal fairness criteria exists. Assume that the strategy results in temporal share vector $w^n$. Let $a^m=\Theta(w^n)$, where the mapping $\Theta(\cdot)$ is defined in Definition \ref{Def:Mapping}. Using Equation \eqref{Eq:RR}, one can construct an ordered round robin strategy satisfying the fairness constraints. The converse proof holds by definition of feasibility.
\end{proof}

The following proposition shows that under equality temporal constraints, the set of feasible window-lengths is non-contiguous at every point.

\begin{Proposition} \label{Prop:Eq-feasibility}
For any scheduling setup with equality constraints $(n,N_{max}, {w}^n, {w}^n, f_{R^m})$, the set of feasible  is characterized as follows:
\begin{align*}
    \mathcal{S}= \Big\{k d_{res}(w^n)\big| k\in \mathbb{N}\Big\},
\end{align*}
where $d_{res}(w^n)= lcm(d_{1},d_2,\cdots,d_n)$, where $w_i= \frac{k_i}{d_i}, i\in [n]$ and $(k_i,d_i)$ are coprime integers. 
\end{Proposition}

\begin{proof}
First we prove that  $\mathcal{S}\subseteq \{k\cdot d_{res}| k\in \mathbb{N}\}$. Assume that $s\in \mathcal{S}({w}^n,{w}^n,N_{max})$, then there exists a scheduling strategy $Q$ satisfying the temporal fairness constraints. Let $l_j, j\in [m]$ be the number of time-slots when the $j$th virtual user is activated by $Q$. Define $a_j= \frac{l_j}{s}, j\in [m]$ as the temporal share of the $j$th virtual user. Then,
\begin{align*}
    \sum_{j: u_i\in \mathcal{V}_j} a_j= \sum_{j: u_i\in \mathcal{V}_j} \frac{l_j}{s}= w_i= \frac{k_i}{d_i}, \forall i\in [n].
\end{align*}
As a result, we must have $d_i | k_is,  i\in [n]$. Using the fact that $(k_i,d_i)$ are coprime, we conclude that $d_i|s, i\in [n]$. Consequently, $d_{res}|s$. This proves that $\mathcal{S}\subseteq \{k d_{res}| k\in \mathbb{N}\}$. In the next step, we will show that for any $s=k\cdot d_{res}, k\in \mathbb{N}$, there exists a scheduling strategy satisfying the equality temporal constraints. Note that since $d_{res}=lcm(d^n)$, we can write $w_i= \frac{k'_i}{k d_{res}}, k'_i\in \mathbb{N}, i\in [n]$. 
It is straightforward to show that the following choice of $Q_t(R^{m\times t})$ satisfies the fairness constraints:
\begin{align*}
    i\in Q_t(R^{m\times t}) \iff \sum_{i'<i}k'_i=\sum_{i'<i}s w_i< t\leq \sum_{i'\leq i}s w_i=\sum_{i'\leq i}k'_i.  
\end{align*}
Note that in the above construction, $Q_t(R^{m\times t})$ is a valid virtual user. In other words, we have $|Q_t(R^{m\times t})|\leq N_{max}$ since $\sum_{i\in [n]}w_i \leq N_{max}$.
\end{proof}
The following theorem shows that under inequality temporal fairness constraints, the set of feasible window-lengths is asymptotically contiguous. 
\begin{Theorem} \label{thm:feas-ineq}
For any scheduling setup with inequality constraints $(n,N_{max}, \underline{w}^n, \overline{w}^n, f_{R^m})$, the set of feasible window-lengths is as follows:
\begin{align}
\label{Eq:Set1}
    \mathcal{S}= \bigcup_{w^n\in \mathcal{W}}\Big\{k d_{res}(w^n)\big| k\in \mathbb{N}\Big\},
\end{align}
where $d_{res}(w^n)$ is defined in Proposition \ref{Prop:Eq-feasibility}, and $\mathcal{W}=\{w^n|\underline{w}^n\leq w^n\leq \overline{w}^n: \sum_{i\in [n]}w_i\leq N_{max}\}$. Particularly, the set is contiguous for large enough window lengths. More precisely, we have:
\begin{align}
\label{Eq:Set2}
    \Big\{d_{res}(\underline{w}^n)+nd_{\epsilon}+k\big|k\in \mathbb{N}\Big\}\subseteq \mathcal{S},
\end{align}
where $d_{\epsilon}=max(d_{\alpha},d_{\delta})$, $\frac{k_{\alpha}}{d_{\alpha}}=N_{max}-\sum_{i\in [n]}\underline{w}_i$, $\frac{k_{\delta}}{d_{\delta}}=min_{i\in [n]}(\overline{w}_i-\underline{w}_i)$, and $(k_{\alpha},d_{\alpha})$, and $(k_{\delta},d_{\delta})$ are coprime integers. 
\end{Theorem}

\begin{proof}
The proof of Equation \eqref{Eq:Set1} follows from Proposition \ref{Prop:Eq-feasibility}. We prove Equation \eqref{Eq:Set2} by construction. Fix $k\in \mathbb{N}$ and let $\underline{w}_i=\frac{l_i}{d_{res}(\underline{w}^n)}, i\in [n]$. Define $w_i(l)=\frac{l_i+nd_{\epsilon}+k-l}{d_{res}(\underline{w}^n)+nd_{\epsilon}+k}, l\in \{0,1,\cdots ,nd_{\epsilon}+k\}$.  In order to have $w^n(l)\in \mathcal{W}$, one needs to verify that $\underline{w}^n\leq w^n(l)\leq \overline{w}^n$ and $\sum_{i\in [n]}w_i(l)\leq N_{max}$. Clearly, we have $\underline{w}_i \leq w_i(0), i\in [n]$ and $w_i(l)$ is decreasing in $l$. Also, $\underline{w}_i > w_i(nd_{\epsilon}+k), i\in [n]$. Consequently, there exists $l_i^*, i\in [n]$ be such that $w_i(l_i^*+1)\leq \underline{w}_i\leq w_i(l^*_i)$. We argue that $\sum_{i\in [n]}w_i(l^*_i)\leq N_{max}$. To see this, note that by construction, we have:
\begin{align*}
    \forall i\in [n]: &w_i(l^*_i)-\underline{w}_i\leq
    w_i(l^*_i)-w_i(l_i^*+1)=
    \frac{1}{d(\underline{w}^n)+nd_{\epsilon}+k}\\
    &\Rightarrow \sum_{i\in [n] }w_i(l^*_i)- \sum_{i\in [n] } \underline{w}_i\leq \frac{1}{d_{\epsilon}}\leq \frac{1}{d_{\alpha}}.
\end{align*}
On the other hand, by definition, we have $\frac{1}{d_{\alpha}}\leq N_{max}-\sum_{i\in [n]}\underline{w}_i$. Consequently, we have shown that $\sum_{i\in [n]}w_i(l^*_i)\leq N_{max}$. Furthermore,
\begin{align*}
   w_i(l^*_i)-\underline{w}_i \leq \frac{1}{nd_{\epsilon}} \Rightarrow 
   \overline{w}_i-w_i(l^*_i)\geq \frac{n-1}{nd_{\epsilon}}\geq 0. 
\end{align*}
Hence, ${w^*}^n= (w_i(l^*_i))_{i\in [n]}\in \mathcal{W}$. This concludes the proof. 
\end{proof} 

\subsection{Characteristics of Optimal Utility}
Next, we investigate the monotonicity of the optimal utility function over the feasible set of window-lengths, and show that in general, it is non-monotonic. This is explained through the following example and the ensuing proposition.

\begin{Example}
\label{Ex:2}
Consider a scheduling system consisting of two users ($n=2$) with $N_{max}=1$. Let the users' utilities be fixed over time and equal to $R_2=2R_1=2$. Consider the temporal constraints $\underline{w}_1=\underline{w}_2=\frac{1}{4}$ and $\overline{w}_1=\overline{w}_2=\frac{3}{4}$. 
Note that $\forall k\geq 2$, $w^2= \big(\frac{\floor{k/2}}{k},\frac{\floor{k/2}}{k}\big)$ is feasible. From Equation \eqref{Eq:Set1},
we have $\mathcal{S}=\{s|s\geq 2\}$. Note that the utility of user $2$ is larger than that of user 1. So, an optimal strategy activates the second user as long as the first user's fairness constraints are not violated. The optimal utility function is given as follows: $U^*_2= \frac{3}{2}$, $U^*_3=\frac{5}{3}$, and $U^*_s=\frac{s+3\floor{s/4}}{s}, s\geq 4$. 
The optimal utility is decreasing when $4|s$ and increasing everywhere else. 
\end{Example}

 Consequently, we have shown the following proposition.

\begin{Proposition}
There exists scheduling setups $(n,N_{max}, \underline{w}^n, $ $\overline{w}^n, f_{R^m})$ for which the optimal utility function $U^*_s$ is non-monotonic in $s$ in an infinite number of points.
\end{Proposition}

The next theorem shows that the optimal utility is a superadditive function of the window-length. This is useful in proving the existence of an optimal utility value as the window-length is increased asymptotically and in analyzing rates of convergence of the optimal utility function to the optimal utility value. 

\begin{Theorem}
\label{Th:Sup}
For any  scheduling setup $(n,N_{max}, \underline{w}^n, $ $\overline{w}^n, f_{R^m})$, the optimal utility function $U^*_s$  is superadditive in $s$:
\begin{align*}
    \forall s,s' \in \mathcal{S}: sU^*_s+s'U^*_{s'}\leq (s+s')U^*_{s+s'}.
\end{align*}
\end{Theorem}
\begin{proof}
 The proof follows by a construction argument. Let $Q^*_s$ and $Q^*_{s'}$ be optimality achieving strategies for window lengths $s$ and $s'$, respectively. Consider the concatenated strategy $Q_{s+s'}$ defined as follows:
\begin{align*}
    Q_{s+s',t}=
    \begin{cases}
    Q^*_{s,t} \qquad & \text{ if } t\leq s\\
    Q^*_{s',t}\qquad & \text{ if } s+1\leq t\leq s+s',
    \end{cases}
\end{align*}
where $Q^*_{s,t}$ and $Q^*_{s',t}$ are the output of the scheduling strategies $Q^*_{s}$ and $Q^*_{s'}$ at time $t$, respectively. The strategy $Q_{s+s'}$ satisfies the temporal fairness constraints and has average utility $\frac{s}{s+s'}U^*_s+\frac{s'}{s+s'}U^*_{s'}$. This concludes the proof. 
\end{proof}
The following is direct consequence of Theorem \ref{Th:Sup}.
\begin{Corollary} \label{Cor:Sup}
For the setup $(n,N_{max}, \underline{w}^n, $ $\overline{w}^n, f_{R^m})$, let $U^*$ be the optimal utility under long-term temporal fairness constraints. Then, $U^*_s\leq U^*, \forall s\in \mathcal{S}$ and $\lim_{s\to \infty} U^*_s=U^*$. 
\end{Corollary}


\section{Multi-user Scheduling Strategy}
\label{Sec:Scheme}
In this section, we provide a scheduling strategy satisfying short-term fairness constraints for an arbitrary feasible window-length. The strategy achieves optimal utility as the window-length is increased asymptotically. We further derive the rates of convergence to the long-term optimal utility. 

It was shown in \cite{arxiv}, that under long-term fairness constraints, a special class of scheduling strategies called threshold based strategies (TBS) achieve optimal system utility. However, application of TBSs does not guarantee short-term fairness. In the following, we modify TBSs to ensure that the fairness criteria are satisfied in the given window-length. At a high level, the proposed strategy modifies the set of virtual users which can be activated at each time-slot based on the virtual users which have been activated in the prior time-slots so as to ensure that the short-term fairness constraints are satisfied. The proposed class of scheduling strategies are called \textit{augmented threshold based strategies} (ATBS). Algorithm \ref{alg:atbs} provides a step-by-step description of the ATBS. The following defines the class of ATBSs.

\begin{Definition}[\bf{ATBS}]
\label{Def:ATBS}
For the scheduling setup $(n,N_{max}, \underline{w}^n, \overline{w}^n, f_{R^m})$ with window-length $s\in \mathcal{S}$, an ATBS is characterized by the vector $\lambda^n\in \mathbb{R}^n$. The strategy $Q_{ATBS}(s,\lambda^n)=(Q_{ATBS,t})_{t\in \mathbb{N}}$ is defined as:
\begin{align}
Q_{ATBS,t}\big(R^{m\times t}\big)=\argmax_{\mathcal{V}_j\in\mathsf{V}_{t}} ~M\big(\mathcal{V}_j,R_{t,j}\big), ~t\in \mathbb{N},
\end{align}
where $M\big(\mathcal{V}_j,R_{t,j}\big)=R_{t,j}+\sum_{i=1}^n \lambda_i \mathbbm{1}_{\{u_i\in\mathcal{V}_{j}\}}$ is the `scheduling measure' corresponding to the virtual user $\mathcal{V}_j$, and $\mathsf{V}_t$ is the 
`feasible virtual user set' at time $t$ and consists of all virtual users $\mathcal{V}_j$ satisfying the following conditions:
\begin{align}
&s-t\geq  \max_{i\in [n]} \left(\ceil{s\underline{w}_i}-(t-1)A^Q_{i,t-1}- \mathbbm{1}_{\{u_i\in \mathcal{V}_j\}}\right), 
\label{Eq:1}
\\
& \floor{s\overline{w}_i}\geq (t-1)A^Q_{i,t-1}+ \mathbbm{1}_{\{u_i\in \mathcal{V}_j\}}, i\in [n],
\label{Eq:2}\\
& (s-t)N_{max} \geq  \sum_{i\in [n]}\left(\ceil{s\underline{w}_i}-(t-1)A^Q_{i,t-1}- \mathbbm{1}_{\{u_i\in \mathcal{V}_j\}}\right)^+,
\label{Eq:3}
\end{align}
where $x^+= x\times \mathbbm{1}_{x\geq 0}$.

\label{def:U_TBS}
\end{Definition}

Inequalities \eqref{Eq:1}--\eqref{Eq:3} ensure short-term fairness. More specifically, \eqref{Eq:1} and  \eqref{Eq:3} verify that if $\mathcal{V}_j$ is activated in the next time-slot, then the lower temporal demands can be satisfied in the remaining time-slots. Inequality \eqref{Eq:2} verifies the feasibility of the upper temporal demands given that $\mathcal{V}_j$ is activated in the next time-slot. 
\begin{Remark}
\label{Rem:sat}
If $t=s$, then Equations \eqref{Eq:1}-\eqref{Eq:3} are equivalent to the temporal fairness constraints in Definition \ref{Def:Strategy}.
\end{Remark}

In order to reduce the computational complexity of the ATBS, the feasible virtual user set may be derived iteratively as in Line 4 of Algorithm \ref{alg:atbs}. This is stated in the next proposition. 

\begin{algorithm}[t]
\caption{Augmented Threshold based Strategy (ATBS)}
\begin{algorithmic}[1]
\STATE $\mathsf{V}_0=\mathsf{V}$
\FOR {$t=1$ to $s$ with step-size $1$ }
    \STATE $\mathsf{V}_t=\phi$
    \FOR {$j: \mathcal{V}_j \in \mathsf{V}_{t-1} $}
        \IF { Equations \eqref{Eq:1}, \eqref{Eq:2}, and \eqref{Eq:3} are satisfied}
            \STATE $ \mathsf{V}_{t} = \mathsf{V}_{t} \cup \{\mathcal{V}_j\}$
        \ENDIF
    \ENDFOR
    \STATE $Q_{ATBS,t}\big(R^{m\times t}\big)=\argmax_{\mathcal{V}_j\in\mathsf{V}_{t}} ~S\big(\mathcal{V}_j,R_{t,j}\big)$ 

\ENDFOR
  \end{algorithmic}
  \label{alg:atbs}
\end{algorithm} 
\begin{Proposition}
\label{Prop:Complexity}
The feasible virtual user set can be derived iteratively since $\mathsf{V}_{t+1}\subset \mathsf{V}_{t}, t\in [s-1]$.
\end{Proposition}

\begin{proof}
 We prove the proposition for inequality \eqref{Eq:1}. The proof for the other two inequalities follows by the same arguments. Assume that inequality \eqref{Eq:1} is violated at time $t-1$:
\begin{align*}
&s-t+1>  \max_{i\in [n]} \left(\ceil{s\underline{w}_i}-(t-2)A^Q_{i,t-2}- \mathbbm{1}(u_i\in \mathcal{V}_j)\right)
\end{align*}
We show that the inequality is also violated at time $t$ regardless of the value of $Q_{ATBS,t}$:
\begin{align*}
&s-t+1-1<  \max_{i\in [n]} \left(\ceil{s\underline{w}_i}-((t-2)A^Q_{i,t-2}+1)- \mathbbm{1}(u_i\in \mathcal{V}_j)\right)\\
&\leq \max_{i\in [n]} \left(\ceil{s\underline{w}_i}-((t-1)A^Q_{i,t-1})- \mathbbm{1}(u_i\in \mathcal{V}_j)\right).
\end{align*}
This completes the proof.
\end{proof}

Next we argue that ATBSs satisfy the short-term temporal constraints. We make use of the following propositions to prove this claim.

\begin{Proposition} \label{Prop:Feas1}
For the setup $(n,N_{max}, \underline{w}^n, \overline{w}^n, f_{R^m})$ and a feasible  window-length $s$, let $Q$ be a strategy satisfying the short term fairness constraints. Then, $Q_t(R^{m\times t})\in \mathsf{V}_1, t\in [s]$. 
\end{Proposition}

\begin{proof}
We prove the proposition by contradiction. Assume that $Q_t(R^{m\times t})\notin \mathsf{V}_1$. Then at least one of the Equations \eqref{Eq:1}-\eqref{Eq:3} is violated.
We continue the proof based on the assumption that Equation \eqref{Eq:1} is violated. The proof when the other cases follows by similar arguments. Assume that
\begin{align*}
    \exists i\in [n]: s-1< \ceil{s\underline{w}_i}- \mathbbm{1}_{\{u_i\in \mathcal{V}_j\}},
\end{align*}
where $\mathcal{V}_j= Q_t(R^{m\times t})$. Then, the temporal share of user $i$ is less than  $ \frac{1}{s}\ceil{s\underline{w}_i}$ even if it is activated in all other time-slots since  $\frac{1}{s}(\ceil{s\underline{w}_i}-1)< \underline{w}_i$.  
This contradicts the assumption that the scheduler satisfies the lower temporal demands for user $i$.
\end{proof}


\begin{Proposition}
\label{prop:sat2}
For the scheduling setup $(n,N_{max}, \underline{w}^n, \overline{w}^n, f_{R^m})$ assume that the window-length $s$ is feasible.  Then, 
\begin{align*}
    Q_{ATBS,t}(R^{m\times t})= \mathcal{V}_j
    \Rightarrow s-t\in \mathcal{S}(\underline{{w}'}^n,\overline{{w'}}^n,\mathsf{V}_{t+1})
\end{align*}
where 
\begin{align*}
    &\underline{w}'_i=\left(\ceil{s\underline{w}_i}- (t-1)A_{i,t-1}^Q- \mathbbm{1}_{\{u_i\in \mathcal{V}_j\}}\right)^+,\quad i\in [n],\\
    &\overline{w}'_i=\floor{s\overline{w}_i}- (t-1)A_{i,t-1}^Q- \mathbbm{1}_{\{u_i\in \mathcal{V}_j\}},\quad i\in [n].
\end{align*}
\end{Proposition}

\begin{proof}
 Note that by definition ${\underline{w}'_i}\leq {\overline{w}'_i}$.
If $Q_{ATBS,t}(R^{m\times t})=\mathcal{V}_j$, then $\mathcal{V}_j\in \mathsf{V}_t$. So,  
\begin{align*}
&s-t\geq  \max_{i\in [n]} (\underline{w}'_i), \qquad \overline{w}'_i\geq 0, i\in [n],
\qquad (s-t)N_{max} \geq  \sum_{i\in [n]}\underline{w}'_i, t\in [s-1].
\end{align*}
Consequently, $s-t\in\mathcal{S}({\underline{w}'}^n,\overline{{w'}}^n,N_{max}) $. So, there exists an ordered round robin strategy $Q_{ORR}$ satisfying the short term fairness constraints $({\underline{w}'}^n,{\overline{w}'}^n)$ with window-length $s-t$. By Proposition \ref{Prop:Feas1}, we have $\{Q_{ORR,t'}|t'\in [s-t]\}\subseteq \mathsf{V}_{t+1}$. Consequently, $s-t\in \mathcal{S}({\underline{w}'}^n,\overline{{w'}}^n,\mathsf{V}_{t+1}).$
\end{proof}
The following theorem shows that ATBSs satisfy the short term fairness constraints for any feasible window-length $s$.  

\begin{Theorem}
For the scheduling setup $(n,N_{max}, \underline{w}^n, \overline{w}^n, f_{R^m})$ assume that the window-length $s$ is feasible. Then, the class of ATBSs described in Definition \ref{Def:ATBS} satisfy the short-term temporal demand constraints.  
\end{Theorem}

\textit{Outline of proof.} From Remark \ref{Rem:sat}, in order for the ATBS to satisfy the short-term temporal constraints, we must have $P(\mathsf{V}_{s}\neq \phi)=1$. Note that by the feasibility assumption, we must have $\mathsf{V}_1\neq \phi$. So, $Q_{ATBS,1}$ is well-defined. 
Then, by Proposition \ref{prop:sat2}, we must have $\mathsf{V}_2\neq \phi$. Similarly, by induction, we must have $\mathsf{V}_s\neq \phi$. 

In \cite{arxiv}, it was shown that TBSs achieve optimal average system utility under long-term temporal fairness constraints. The following theorem builds upon this to show that ATBSs achieve optimal average system utility as the window-length is increased asymptotically and to
provide rates of convergence to the optimal long-term utility.

\begin{Theorem} \label{Th:Conv}
For the scheduling setup $(n,N_{max}, \underline{w}^n, \overline{w}^n, f_{R^m})$ let ${\lambda^*}^n$ be the threshold vector which achieves optimal performance under long-term fairness constraints\footnote{The existence of ${\lambda^*}^n$ is shown in \cite{arxiv}.}. Then, 
\begin{align*}
    1- \frac{m}{4s\epsilon^2}
    \leq\lim_{s\to \infty} \frac{U_{ATBS}(s,{\lambda^*}^n)}{U_{TBS}}
    \leq
    \lim_{s\to \infty} \frac{U^*_s}{U_{TBS}}
    \leq 1,
\end{align*}
where $U_{TBS}$ is the optimal expected utility under long-term fairness constraints and ${w^*}^n$ is the corresponding temporal demand vector, $U_{ATBS}(s,{\lambda^*}^n)$ is the expected utility due to $Q_{ATBS}(s,{\lambda^*}^n)$, $\epsilon= \min(\epsilon_1,\epsilon_2)$, $\epsilon_1= \min_{i\in [n]}(w^*_i-\underline{w}_i,\overline{w}_i-w^*_i)$, $\epsilon_2=N_{max}-\sum_{i\in[n]}{w^*_i}$.  
\end{Theorem}
\begin{proof}
Let $M_{t}, t\leq [s]$ be the index of the virtual user which is selected by the ATBS strategy at time $t$. Let $A$ be the stopping time for which $\mathsf{V}_{A-1}=\mathsf{V}\neq \mathsf{V}_A$.  As mentioned in \cite{liu-jsac}, the sequence $M^A$ consists of independently and identically  distributed components with $Pr(M_t=j)=a_j, j\in [m], t\in [A]$, where $a^m= \Theta({w^*}^n)$, and $\Theta(\cdot)$ is defined in Definition \ref{Def:Mapping}. For a given natural number $s'\leq A$, we define the $\epsilon$-typical set of sequences as:
\begin{align*}
    \mathcal{A}^{s'}_\epsilon(M)=\left\{M^{s'} \Big ||\frac{1}{s'}\sum_{i\in [s']}\mathbbm{1}_{\{M_j=j\}}-a_j|\leq \epsilon, \forall j\in [m]\right\}.
\end{align*}
From \cite{csiszar2011information}, we have $Pr(\mathcal{A}^{s'}_\epsilon(M))\geq 1- \frac{m}{4s'\epsilon^2}$. Note that if $\epsilon$ is taken as in the theorem statement, and if $M^s\in \mathcal{A}^{s}_\epsilon(M)$, then $\mathsf{V}_{s}=\mathsf{V}$. On the other hand, from Wald's identity \cite{janssen2006applied} we have 
\begin{align}
\label{Eq:Low}
U^*_s&\geq U_{ATBS}(s,{\lambda^*}^n)=\mathbb{E}\left(\frac{1}{s}\sum^s_{k=1}\sum_{j=1}^m R_{j,k}\mathbbm{1}_{\big\{Q_{ATBS,k}(R^{m\times k})=\mathcal{V}_{j}\big\}}\right)
\\&\geq\mathbb{E}\left(\frac{1}{s}\sum^A_{k=1}\sum_{j=1}^m R_{j,k}\mathbbm{1}_{\big\{Q_{ATBS,k}(R^{m\times k})=\mathcal{V}_{j}\big\}}\right)
=\frac{1}{s}\mathbb{E}(A)\mathbb{E}\left(\sum_{j=1}^m R_{j,1}\mathbbm{1}_{\big\{Q_{TBS,1}(R^{m})=\mathcal{V}_{j}\big\}}\right)= \frac{1}{s}\mathbb{E}(A) U^*,   \nonumber
\end{align}
where we have used the fact that $U^*=U_{TBS}$ when ${\lambda^*}^n$ is used as the threshold vector. Consequently,
\begin{align*}
    U^*_s\geq \frac{1}{s}\mathbb{E}(A) U^*
    \geq Pr(A=s) U^*
    \geq  (1- \frac{m}{4s\epsilon^2})U^*.
\end{align*}
\end{proof}

\begin{figure}[h]
 \centering \includegraphics[width=0.60\linewidth, draft=false]{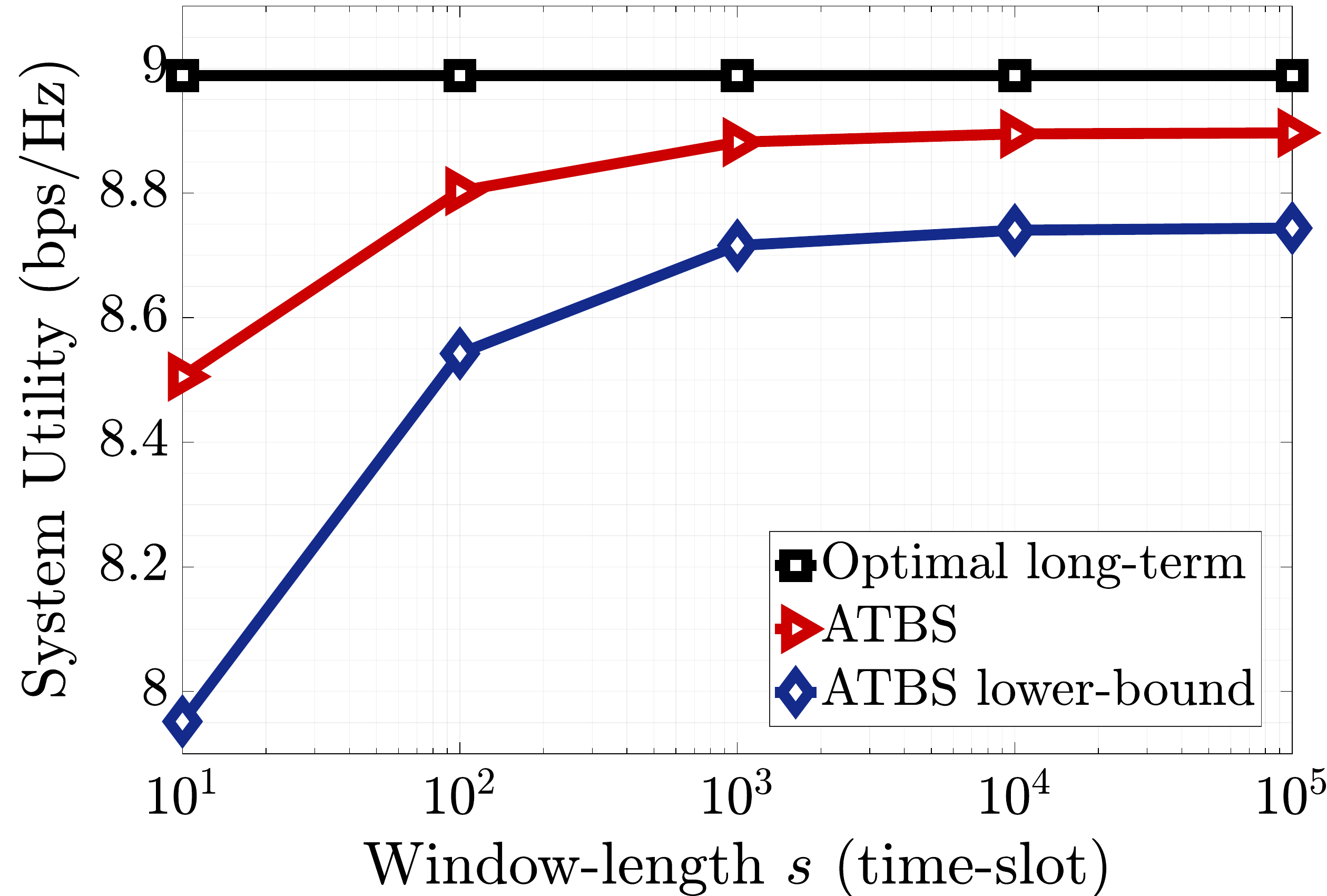}
  \caption{System utility as a function of window-length $s$ for the optimal threshold based strategy (TBS) and proposed augmented threshold based strategy (ATBS) using the same thresholds as the optimal TBS.}
 \label{fig:utility}
 \vspace{-15pt}
\end{figure} 

\section{Simulation Results}
\label{Sec:Sim}
In this section, we provide simulations of practical scenarios to validate the results presented in the previous sections and evaluate the performance short-term fair scheduling strategy proposed in Section \ref{Sec:Scheme}. We consider the downlink of a single-cell wireless system consisting of a BS and five users distributed uniformly at random in a ring around the BS with inner and outer radii of 20 m and 100 m, respectively. We consider a practical propagation channel model including path loss, shadowing, and Rayleigh fading as well as a truncated Shannon rate model provided in Table I of \cite{arxiv}. The performance value of the virtual users are calculated as in \cite{arxiv} where the users perform successive interference cancellation to decode their signals and achieve maximum symmetric BC sum-rate. We assume that at most two users can be activated at each time-slot, i.e.  $N_{max}=2$, and $\underline{w}_i=0.2, \overline{w}_i=1, \forall i\in [5]$, i.e., each user is active for at least one-fifth of the time-slots. Furthermore, we consider fairness window-lengths $s\in\{10,10^2,10^3,10^4,10^5\}$. From Theorem \ref{thm:feas-ineq}, it can be shown that these window-lengths are feasible.


Figure \ref{fig:utility} shows the average system utility as a function of fairness window-length $s$. From \cite{arxiv}, it is known that optimal system utility under long-term temporal constraints is achieved by TBSs. By Corollary \ref{Cor:Sup}, the  optimal utility $U^*$ provides an upper-bound for the utility of ATBS. We observe that the utility of ATBS, using the same thresholds as the optimal TBS, approaches to the optimal utility as window-length increases, confirming Theorem \ref{Th:Conv}. Furthermore, we can see that the gap between the utility of ATBS and optimal utility is small even for relatively small window-length such as $s=100$.  Figure \ref{fig:utility} also plots the simulated value of the lower-bound provided in the right hand side of Equation \eqref{Eq:Low}.


\section{Conclusion}
\label{Sec:Conc}
We have investigated multi-user scheduling under short-term temporal fairness constraints. We have characterized the set of feasible window lengths as a function of system parameters, and we have shown that the optimal system utility is non-monotonic and super-additive in window-length. We have proposed a scheduling strategy which satisfies short-term fairness constraints for arbitrary feasible window-lengths. The utility due to the proposed strategy approaches long-term optimal utility as the window-length is increased asymptotically.  Numerical simulations are provided in various practical scenarios to verify the theoretical results. An avenue of future research is to consider scheduling for multi-cell wireless networks under short-term and long-term fairness constraints.

\bibliographystyle{IEEEtran}
\bibliography{reference}

\end{document}